%% file: main.tex
\documentclass[11pt]{article}

\usepackage[margin=1.1in]{geometry}
\usepackage{amsmath,amssymb,amsthm}
\usepackage{booktabs}
\usepackage{graphicx}
\usepackage{microtype}
\usepackage{xcolor}
\usepackage{algorithm}
\usepackage{algpseudocode}
\usepackage[numbers,sort&compress]{natbib}
\usepackage[colorlinks=true,linkcolor=blue!50!black,citecolor=blue!50!black,urlcolor=blue!50!black]{hyperref}

\newcommand{\simplex}{\Delta^{K-1}}
\newcommand{\E}{\mathbb{E}}
\newtheorem{proposition}{Proposition}
\newcommand{\evalyears}{13.6}

\title{\textbf{KellyBoost: Growth-Optimal Portfolio Construction\\ with Gradient-Boosted Trees}}
\author{Jiayu Li\\[0.3em]\normalsize\texttt{lijiayu2027@outlook.com}}
\date{August 2026}

\begin{document}
\maketitle

\begin{abstract}
End-to-end portfolio learning --- mapping features directly to portfolio
weights under a financial objective, with no intermediate return forecast ---
has so far required backpropagation, shutting out gradient-boosted decision
trees, the workhorse of tabular machine learning. We build the missing
formulation. \textsc{KellyBoost} is a single multi-output XGBoost model
whose softmax output \emph{is} the portfolio: with $y$ the vector of
per-asset holding-period returns, the training loss is $-\log(1 + w^{\top}y)$
--- the negative log growth rate, so the fitted model is the growth-optimal
(Kelly) allocation conditioned on the features. The objective is exact
rather than a surrogate: we derive the gradient, the analytic diagonal
Hessian and the full Hessian in closed form, verify them by finite
differences, and ship a dependency-free reference engine. On a 23-year,
eight-asset public-data testbed under a strictly separated
selection/estimation protocol, we test the objective with a $2 \times 2$
design --- \{boosted tree, MLP\} $\times$ \{growth objective,
classification surrogate\} --- and the growth objective improves deployed
log growth over the surrogate within \emph{both} learner classes, in both
of two independently built feature pipelines: four of four cells ---
and, because the growth objective also trades less in every cell,
transaction costs widen the gap rather than eroding it. We are
equally explicit about the objective's boundary, because the same testbed
measures it: faithfully optimizing an \emph{estimated} conditional
full-Kelly allocation concentrates aggressively, and in this regime the
end-to-end learners trail the classic two-stage predict-then-optimize
pipeline, whose squared-error forecasts shrink toward zero and whose
optimizer therefore holds mild portfolios.
The exact Hessian beats its first-order substitutes on the selection
protocol by a factor of three and loses to them deployed: the optimizer is
not the problem; the estimated full-Kelly target is. Code, data and
scripts regenerate every number from one committed CSV without an API key.
\end{abstract}

\section{Introduction}\label{sec:intro}

Portfolio construction is a decision problem, but most machine-learning
treatments of it are prediction problems with a decision bolted on: a model
is trained to forecast returns under a statistical loss, and the forecasts
are then handed to an optimizer, together with a separately estimated
covariance matrix \citep{markowitz1952}. The mismatch is well documented.
Small forecast errors are amplified by the optimizer in exactly the
directions the optimizer cares about; expected returns are the quantity
financial data identifies \emph{worst} \citep{merton1980}; and the two
stages optimize different objectives, so improving the first need not
improve the second \citep{elmachtoub2022}.

End-to-end (decision-focused) learning removes the seam: train the model
directly on the decision objective. In portfolio management this line of
work has been carried almost exclusively by neural networks
\citep{zhang2020deep,zhang2021universal,uysal2024}, which directly optimize
Sharpe-ratio or utility objectives by gradient descent --- because gradient
descent is what the formulation requires. Gradient-boosted decision trees
(GBDTs), the workhorse of tabular machine learning and of quantitative
practice \citep{grinsztajn2022,gu2020}, have been shut out of this line of
work entirely: a tree ensemble is not trained by backpropagation, so a
practitioner who wants boosted trees on an allocation problem has had to
route them through a surrogate --- forecast returns, or classify the best
asset --- and accept the seam. What has been missing is the bridge: a
tree-boosting formulation of the end-to-end allocation problem, so that
the choice of learner class and the choice of training objective stop
being one decision.

This paper builds that bridge, and it turns out to require nothing
approximate. Whether the bridge is \emph{worth crossing} --- whether the
decision objective beats the surrogates a practitioner would otherwise use
--- is an empirical question, and answering it cleanly, in both
directions, is the paper's second half. Our contributions:

\begin{enumerate}
\item \textbf{An exact growth-optimal objective for boosted trees}
(Section~\ref{sec:method}). One multi-output XGBoost model
\citep{chen2016xgboost} produces a logit vector per row; its softmax is a
long-only, fully-invested portfolio on the simplex; the per-row loss is
$-\log(1+w^{\top}y)$, whose sample mean is the negative log growth rate ---
the Kelly criterion \citep{kelly1956,breiman1961}. We derive the gradient
and the analytic diagonal Hessian in closed form
(verified by finite differences in the accompanying test suite), so the model
trains with true second-order boosting steps. The tuning metric is the
training loss itself: nothing in the pipeline optimizes a proxy.
\item \textbf{A deployment recipe for tree instability}
(Section~\ref{sec:ensemble}). We show that two boosted fits whose training
sets differ by a \emph{single dropped row} can disagree materially in the
allocation they produce --- near-tied split points flip, and the
perturbation saturates at one row. A single fit is a draw from a
distribution, not its centre. We deploy a leave-one-out ensemble: $K$
members, each fitted on the data minus one row, at \emph{unchanged}
hyperparameters --- unlike bagging, the perturbation is too small to change
effective capacity, so parameters tuned for a single fit remain valid.
\item \textbf{A cleanly separated evaluation protocol}
(Section~\ref{sec:protocol}). Hyperparameters are selected on a development
segment (ending 2012) with a purged single-row-block walk-forward whose score
is treated strictly as a \emph{selection signal}; performance is then
measured once, on the untouched evaluation segment (2013--2026), by a
deployment walk-forward with frozen parameters, expanding windows, and an
end-anchored decision grid. The distinction between selection signals and
performance estimates is enforced throughout.
\item \textbf{A fully reproducible testbed} (Section~\ref{sec:data}). Eight
asset legs, a 7{,}871-column candidate feature set, and a joint
(hyperparameter, feature list) search shared by every learned method
(Section~\ref{sec:search}), over 2003--2026 --- built exclusively from
freely downloadable public market data (no API keys), with causality
guaranteed by construction: every input is a same-day market close, every
transform is backward-looking, and the raw snapshot is committed to the
repository.
\end{enumerate}

Out of sample, over \evalyears{} years, the comparisons that isolate the
loss all land the same way. The study's $2 \times 2$ design --- \{boosted
tree, MLP\} $\times$ \{growth objective, classification surrogate\} ---
holds the learner class fixed and swaps only the training loss, and the
growth objective wins within both learner classes, in both feature
pipelines: four of four cells, with a pooled paired effect of $+0.18$ per
20-day decision ($\times 100$) that \emph{rises} to $+0.25$ under 20 bps
of transaction costs, because the growth objective also trades less in
every cell (Section~\ref{sec:results}). That is the
paper's claim, and it is the whole claim. Comparisons \emph{across}
learner classes (tree versus MLP at fixed loss) are confounded by
capacity and optimization geometry, land within bootstrap noise, and we
draw no conclusion from them. The comparisons against
differently-\emph{shaped} pipelines mark the objective's boundary rather
than its value: the two-stage predict-then-optimize pipeline beats
every end-to-end learner in this regime (no single
pairwise gap clears its bootstrap interval --- thirteen years of monthly
decisions is a small sample, and we say so). The ablations give that
boundary its mechanism: the exact Hessian wins the development-segment
selection score by a factor of three over a constant-curvature substitute
and then \emph{loses} to it out of sample (Section~\ref{sec:ablations}).
Optimization skill transfers; the estimated full-Kelly target does not.
Every mechanism in the study that shrinks allocations toward
diversification --- the two-stage pipeline's near-zero squared-error
forecasts, the constant Hessian's undersized steps --- improves deployed performance, which locates the boundary
squarely in the growth-optimal objective's interaction with estimation
error \citep{maclean2011,demiguel2009}, not in trees, boosting, or the
derivations. Our code, data and scripts regenerate every table and
figure from one committed CSV.

\section{Related work}\label{sec:related}

\paragraph{Growth-optimal investment.} The criterion of maximizing the
expected logarithm of wealth originates with \citet{kelly1956} and was
given its asymptotic optimality theory by \citet{breiman1961}; see
\citet{maclean2011} for a survey and \citet{hakansson1971} for its relation
to mean-variance analysis. \citet{cover1991} constructs portfolios that are
growth-optimal in hindsight without distributional assumptions. This
literature conditions on \emph{no} information or on price history alone;
KellyBoost is the conditional version: a nonparametric estimate of the
log-optimal portfolio as a function of an arbitrary feature vector.

\paragraph{End-to-end portfolio learning.} \citet{zhang2020deep} train
networks to maximize Sharpe ratio directly; \citet{zhang2021universal}
generalize to several objectives while bypassing covariance estimation;
\citet{uysal2024} do end-to-end risk budgeting. All are neural. The
predict-then-optimize critique is formalized by \citet{elmachtoub2022}.
Our contribution is orthogonal to architecture novelty: we bring the
end-to-end objective to the model class that empirically dominates tabular
problems of this size \citep{grinsztajn2022}.

\paragraph{Boosting with custom objectives.} Gradient boosting accepts any
twice-differentiable loss \citep{friedman2001,chen2016xgboost}; XGBoost's
multi-output mode grows one tree per round with vector-valued leaves. We
use this machinery unchanged --- the novelty is the objective and the
demonstration that its \emph{exact} curvature, not a first-order
substitute, is what makes it work (Section~\ref{sec:ablations}).

\paragraph{Evaluation discipline.} Purging and embargoing overlapping
labels follows \citet{lopezdeprado2018}; the reasons to distrust selected
backtest maxima are quantified by \citet{bailey2014}; sensitivity of
backtest conclusions to the rebalance grid's phase is documented by
\citet{hoffstein2020}, which motivates our end-anchored decision grid.

\section{Method}\label{sec:method}

\subsection{The objective and its exact derivatives}\label{sec:objective}

Let each training row consist of features $x_t \in \mathbb{R}^d$ and
realized simple returns $y_t \in \mathbb{R}^K$ of $K$ assets over the
holding period $[t, t+h]$. A multi-output boosted-tree model produces
logits $z(x_t) \in \mathbb{R}^K$; the portfolio is
\[
w_t \;=\; \sigma(z(x_t)) \in \simplex,
\qquad
\sigma_k(z) = \frac{e^{z_k}}{\sum_j e^{z_j}},
\]
i.e.\ long-only and fully invested by construction (a cash leg makes
``fully invested'' unrestrictive in practice). The portfolio's
holding-period return is $S_t = w_t^{\top} y_t$ and the per-row loss is
\begin{equation}
\ell_t \;=\; -\log(1 + S_t),
\qquad
\frac{1}{n}\sum_t \ell_t \;=\; -\,\widehat{\E}\!\left[\log(1+S)\right],
\label{eq:loss}
\end{equation}
the negative sample log growth rate: minimizing \eqref{eq:loss} maximizes
the geometric growth of wealth, so the fitted model is the growth-optimal
(Kelly) allocation conditioned on $x$. There is no forecasting step, no
covariance matrix, and no separate optimizer; and because \eqref{eq:loss}
is itself the economic objective, hyperparameter selection and evaluation
can use the same quantity, eliminating every surrogate-metric seam in the
pipeline.

The population version of this statement pins down what the model estimates:

\begin{proposition}[conditional growth optimality]\label{prop:kelly}
Let $(x, y)$ be jointly distributed with $\E\,|\log(1 + w^{\top}y)| < \infty$
for all $w \in \simplex$. Among all measurable maps $z : \mathbb{R}^d \to
\mathbb{R}^K$, any minimizer of $\E\left[-\log\!\big(1 +
\sigma(z(x))^{\top} y\big)\right]$ satisfies, for almost every $x$,
\[
\sigma(z(x)) \;\in\; \arg\max_{w \in \overline{\mathrm{im}\,\sigma}}\;
\E\left[\log(1 + w^{\top}y) \,\middle|\, x\right],
\]
i.e.\ the fitted portfolio is the growth-optimal (Kelly) portfolio on the
simplex, conditioned on the features.
\end{proposition}

\begin{proof}
The objective is an expectation of a per-$x$ quantity, and $z(x)$ can be
chosen freely for each $x$; hence the minimization decouples pointwise:
$z^{*}(x)$ must maximize $\E[\log(1 + \sigma(z)^{\top}y) \mid x]$ over
$z \in \mathbb{R}^K$, whose image under $\sigma$ is dense in $\simplex$.
\end{proof}

\noindent KellyBoost is thus a nonparametric, tree-structured estimator of
the conditional Kelly allocation --- the conditional counterpart of the
unconditional growth-optimal portfolios of the classical literature
\citep{kelly1956,breiman1961,cover1991}, with the boosting machinery
supplying the conditioning.

XGBoost requires the gradient and Hessian of $\ell$ with respect to the raw
scores $z$. Both are available in closed form. Using
$\partial \sigma_k / \partial z_k = \sigma_k(1-\sigma_k)$ and
$\partial S / \partial z_k = \sigma_k (y_k - S)$:
\begin{equation}
g_k \;\equiv\; \frac{\partial \ell}{\partial z_k}
\;=\; \frac{\sigma_k \,(S - y_k)}{1 + S},
\qquad
h_k \;\equiv\; \frac{\partial^2 \ell}{\partial z_k^2}
\;=\; g_k\,(1 - 2\sigma_k) \;+\; g_k^2 .
\label{eq:gradhess}
\end{equation}
The derivation is three lines (Appendix~\ref{app:derivation}); both
expressions are verified against central finite differences in the
accompanying test suite. The gradient has a transparent reading: asset $k$'s logit is
pushed up when it out-performs the portfolio it is competing against
($y_k > S$), with force proportional to its current weight and inversely
proportional to realized wealth growth --- the compounding term $1/(1+S)$
that distinguishes log growth from mean return.

The loss \eqref{eq:loss} is not convex in $z$, so the true diagonal
curvature $h_k$ can be negative; XGBoost's leaf weights divide by the
aggregated Hessian, which must be positive, and we take $|h_k|$. The next
subsection examines this safeguard --- and the diagonal approximation
itself --- in detail, because neither is cosmetic.

\paragraph{A risk-aversion dial.} Nothing above is specific to the
logarithm. The accompanying implementation accepts the CRRA family
$\phi_{\gamma}(S) = \big((1+S)^{1-\gamma} - 1\big)/(\gamma - 1)$, with
$\gamma \to 1$ recovering $-\log(1+S)$: the gradient is
$\phi_{\gamma}'(S)\, a_k$ and the analytic diagonal Hessian
$\phi_{\gamma}''(S)\, a_k^2 + g_k(1 - 2\sigma_k)$, where
$a_k = \sigma_k (y_k - S)$ (Appendix~\ref{app:derivation}), both verified
by finite differences. Training with $\gamma > 1$ tempers the
growth-optimal policy inside the objective, the utility-side analogue of
fractional Kelly \citep{maclean1992}. The hand-built pipeline fixes
$\gamma = 1$; the joint search of Section~\ref{sec:search} is offered
$\gamma \in [1, 10]$ as a tunable and commits $\gamma = 1.14$ --- an
essentially logarithmic setting: the selection protocol sees no
in-sample reason to temper the objective, which is why
Section~\ref{sec:limitations} argues the risk dial must be set by
preference rather than by search.

\subsection{The curvature, examined}\label{sec:curvature}

Three questions a second-order method on a non-convex loss owes its
reader: does the rectified step still descend, how often does
rectification actually fire, and what do the discarded off-diagonal terms
cost? Each has a sharp answer here, and all measurements below are on the
real training problem at the frozen parameters
(\texttt{experiments/curvature.py}).

\paragraph{Descent survives rectification.} A leaf's value under any
\emph{positive} per-coordinate curvature $\tilde H_{jk} > 0$ is
$v_{jk} = -G_{jk}/(\tilde H_{jk} + \lambda)$, so the first-order change of
the round objective is $\sum_k G_{jk} v_{jk} = -\sum_k G_{jk}^2 /
(\tilde H_{jk} + \lambda) < 0$: every leaf update is a descent step on the
current linearization regardless of how the curvature was made positive,
with $\lambda$ bounding the step. Rectification therefore affects step
\emph{sizes}, never step \emph{direction}, and the learning rate $\eta$
provides the usual damping on top.

\paragraph{Rectification is the working regime, not an edge case.}
Holding-period returns are small, so expanding \eqref{eq:gradhess} in
$\|y\|$ gives $h_k = g_k(1 - 2\sigma_k) + g_k^2$ with the linear term
dominant: the exact curvature's sign follows the gradient's whenever
$\sigma_k < \tfrac12$, i.e.\ negative curvature is expected on roughly the
cells whose asset is currently \emph{outperforming} the portfolio.
Measured at the committed configuration: $50\%$ of (row, asset) cells are
rectified at the first round, declining to $27\%$ by the last. The safeguard is
central to the objective, which is why it deserves a name rather than an
apology: $|h_k|$ is precisely the diagonal case of saddle-free Newton,
which preconditions by $|H|$ in eigenvalue terms \citep{dauphin2014}. The
ablation's \texttt{grad2} mode is the other principled positive surrogate
--- the diagonal Gauss--Newton curvature $\phi'' a_k^2$
\citep{martens2010}, which for the log loss equals $g_k^2$ exactly --- so
Table~\ref{tab:ablations} compares the exact rectified diagonal against
\emph{both} textbook alternatives, not against a straw man.

\paragraph{The off-diagonal terms are measurable --- and priced in.} The
full per-row Hessian has the closed form
\begin{equation}
H \;=\; \phi''(S)\, a a^{\top}
\;+\; \phi'(S)\,\big(\mathrm{diag}(a) - \sigma a^{\top} - a \sigma^{\top}\big),
\qquad a = \big(\mathrm{diag}(\sigma) - \sigma \sigma^{\top}\big)\, y,
\label{eq:fullhess}
\end{equation}
symmetric by construction (Appendix~\ref{app:derivation}) and
finite-difference-verified. Off-diagonal entries carry about half its
Frobenius mass per row ($0.50$ at round $0$, drifting to $0.62$), so the
diagonal approximation cannot be defended by sparsity. It is defended at
the operating point instead: what a leaf applies is not the row Hessian
but $\sum_{i \in I_j} H_i$ shrunk by the tuned $\lambda$, and there the
diagonal step $-G/(\Sigma|h| + \lambda)$ and the full eigenvalue-rectified
Newton step $-(|\Sigma H| + \lambda I)^{-1} G$ are directionally
indistinguishable --- mean cosine $\geq 0.99$ on leaf-sized row sets at
every stage of training. The end-to-end check agrees: the pure-numpy
engine optionally solves every leaf against the aggregated full Hessians
(\texttt{leaf\_solver="full"}; split search unchanged), and under the
tuning protocol the two leaf solvers score within noise of each other
(Section~\ref{sec:ablations}). The conclusion worth carrying forward:
with $(\lambda, \texttt{min\_child\_weight})$ at their tuned values, the
curvature's real contribution is \emph{per-coordinate step scaling} ---
which is exactly the margin the Hessian-mode ablation prices, and why the
all-ones Hessian, whose scale is wrong by orders of magnitude, fails
hardest there.

\subsection{Second-order boosting with vector leaves}\label{sec:trees}

The model is the standard boosted additive ensemble, except that each tree
is vector-valued. After $M$ rounds,
\[
z(x) \;=\; \eta \sum_{m=1}^{M} f_m(x),
\qquad
f_m : \mathbb{R}^d \to \mathbb{R}^K,
\]
where each $f_m$ is a binary decision tree whose \emph{leaves hold
$K$-vectors} (XGBoost's \texttt{multi\_strategy = multi\_output\_tree}): a
row is routed by ordinary scalar feature splits to a single leaf $j$, and
the whole logit vector receives that leaf's value $v_j \in \mathbb{R}^K$.

\paragraph{One round.} Let $Z$ be the current logit matrix and
$g_{ik}, h_{ik}$ the per-row derivatives \eqref{eq:gradhess} evaluated at
$Z$. Because the Hessian we supply is diagonal, the second-order expansion
of the loss around $Z$ separates over outputs:
\begin{equation}
\sum_i \ell\big(z_i + f(x_i)\big)
\;\approx\;
\text{const} + \sum_{i}\sum_{k}
\Big[ g_{ik}\, f_k(x_i) + \tfrac12\, h_{ik}\, f_k(x_i)^2 \Big]
\;+\; \gamma\, T \;+\; \tfrac{\lambda}{2} \sum_{j,k} v_{jk}^2,
\label{eq:round}
\end{equation}
with $T$ the number of leaves. For a \emph{fixed} tree structure with leaf
row-sets $I_j$, writing $G_{jk} = \sum_{i \in I_j} g_{ik}$ and
$H_{jk} = \sum_{i \in I_j} h_{ik}$, the optimal leaf vector and its
objective reduction are the coordinate-wise Newton step
\begin{equation}
v_{jk}^{*} \;=\; -\,\frac{G_{jk}}{H_{jk} + \lambda},
\qquad
\text{reduction}(I_j) \;=\; \frac12 \sum_{k=1}^{K}
\frac{G_{jk}^2}{H_{jk} + \lambda}.
\label{eq:leaf}
\end{equation}
Split search is the usual greedy scan, but scored jointly: splitting $I$
into $(I_L, I_R)$ gains
\begin{equation}
\mathrm{Gain}
= \frac12 \sum_{k=1}^{K} \left[
\frac{G_{Lk}^{2}}{H_{Lk}+\lambda}
+ \frac{G_{Rk}^{2}}{H_{Rk}+\lambda}
- \frac{G_{k}^{2}}{H_{k}+\lambda}
\right] - \gamma .
\label{eq:gain}
\end{equation}
Equations \eqref{eq:leaf}--\eqref{eq:gain} make the division of labour
explicit: \emph{given} the partition, the $K$ outputs decouple into $K$
independent scalar Newton problems; all coupling between assets flows
through the \emph{shared split search}, where a candidate split must pay
for itself summed across every output. A vector-leaf tree is exactly $K$
scalar boosting problems forced to agree on one partition of feature
space. Algorithm~\ref{alg:kellyboost} assembles the whole procedure; the
per-round cost is that of a scalar tree plus a factor $K$ in the leaf
statistics, and the node Hessian mass that \eqref{eq:leaf} accumulates is
also what the minimum-child-weight constraint meters, so the familiar
regularizers carry over unchanged.

\begin{algorithm}[t]
\caption{KellyBoost training (exact objective, vector leaves)}
\label{alg:kellyboost}
\begin{algorithmic}[1]
\Require features $X \in \mathbb{R}^{n \times d}$; holding-period returns
$Y \in \mathbb{R}^{n \times K}$; rounds $M$; learning rate $\eta$;
regularizers $\lambda, \gamma$
\State $Z \gets 0_{n \times K}$ \Comment{logit matrix}
\For{$m = 1, \dots, M$}
  \State $\sigma_i \gets \mathrm{softmax}(Z_{i\cdot})$;\quad
         $S_i \gets \sigma_i^{\top} Y_{i\cdot}$ \textbf{for all} $i$
  \State $g_{ik} \gets \sigma_{ik}\,(S_i - Y_{ik})\,/\,(1 + S_i)$
         \Comment{exact gradient, eq.~\eqref{eq:gradhess}}
  \State $h_{ik} \gets \big|\, g_{ik}(1 - 2\sigma_{ik}) + g_{ik}^2 \,\big|$
         \Comment{analytic diagonal Hessian, abs safeguard}
  \State grow one tree $f_m$ greedily from the root:
  \Statex \hspace{\algorithmicindent}\hspace{\algorithmicindent}
    at each node, take the feature/threshold split maximizing
    $\mathrm{Gain}$ in eq.~\eqref{eq:gain};
  \Statex \hspace{\algorithmicindent}\hspace{\algorithmicindent}
    stop when no split has $\mathrm{Gain} > 0$ or a size/Hessian-mass
    constraint binds
  \State set each leaf $j$ of $f_m$ to the $K$-vector
         $v_{jk} = -\,G_{jk}/(H_{jk} + \lambda)$
         \Comment{Newton step, eq.~\eqref{eq:leaf}}
  \State $Z \gets Z + \eta\, f_m(X)$
\EndFor
\State \Return $z(\cdot) = \eta \sum_m f_m(\cdot)$;\quad
       portfolio $w(x) = \mathrm{softmax}(z(x))$
\end{algorithmic}
\end{algorithm}

\paragraph{Why one tree per round fits allocation.} The softmax couples
the outputs --- only differences between logits matter --- so allocation
is inherently a joint decision (``growth stock \emph{versus} gold''), and
a shared partition of feature space is both a useful inductive bias and a
$K$-fold parameter reduction: a regime variable that matters for the whole
book (a curve inversion, a volatility spike) buys one split serving all
$K$ outputs, and the joint gain \eqref{eq:gain} prices exactly that. The
alternative, $K$ independent trees per round
(\texttt{one\_output\_per\_tree}), can only rediscover the same split $K$
times from $K$ noisy scalar signals. The ablation
(Table~\ref{tab:ablations}) quantifies the difference. A version note
that matters for reproduction: xgboost $\geq 3.4$ is required, as earlier
versions silently ignore \texttt{min\_child\_weight} under vector-leaf
training.

\paragraph{A dependency-free reference implementation.} The accompanying code
contains two interchangeable engines behind one interface: XGBoost's C++
vector-leaf trees (used for every number in this paper), and a
self-contained pure-numpy implementation of Algorithm~\ref{alg:kellyboost}
of about two hundred lines --- quantile binning, leaf-wise growth, the
joint gain \eqref{eq:gain} and vector Newton leaves \eqref{eq:leaf} ---
with no compiled dependency. Its split selection and leaf values are
tested \emph{exactly} (to $10^{-10}$) against brute-force evaluations of
\eqref{eq:leaf}--\eqref{eq:gain}, and the two engines agree out of sample
on planted-signal problems. Building it surfaced a portability caveat
worth recording: \texttt{min\_child\_weight} semantics under vector leaves
are not standardized --- our engine meters a child's Hessian mass summed
over outputs, XGBoost's gate is stricter --- so the value buys different
effective capacity on different engines and should be re-tuned, not
ported.

\subsection{Deployment: fixed-parameter leave-one-out ensembling}
\label{sec:ensemble}

Boosted trees are chaotically sensitive to the training set: when two
candidate splits are near-tied, an infinitesimal perturbation flips the
winner, and the divergence compounds through subsequent rounds. In our
setting the effect is large enough to matter economically. Deleting a
\emph{single} row from a multi-thousand-row training set --- roughly
$0.02\%$ of the data, and an interior row whose information is nearly
duplicated by its neighbours, since adjacent rows share $h-1$ days of their
labels --- moves the allocation the model produces for the live decision
row by several percentage points on a leg
(measured in Section~\ref{sec:ablations}). A single fit is therefore one draw from
a distribution induced by irrelevant detail, not the centre of that
distribution.

The deployment recipe: fit $K_{\mathrm{ens}}$ members, member $s$ on the
training set minus one row drawn deterministically from seed $s$, and
average the members' weight vectors (the mean of simplex points stays on
the simplex). Crucially, every member runs the \emph{committed}
hyperparameters unchanged. Classic bagging \citep{breiman1996bagging}
resamples aggressively and would change the effective capacity the
hyperparameters were tuned for; the one-row perturbation is instead the
smallest change that decorrelates the split-tie coin flips, so tuning
remains valid and nothing needs re-searching. Averaging pulls the
allocation toward the centre at the usual $1/\sqrt{K_{\mathrm{ens}}}$ rate,
and the effect saturates quickly (Section~\ref{sec:ablations}); we deploy
$K_{\mathrm{ens}} = 4$ for every learned method in the comparison. Seeds
are the member indices, so the deployed book is a deterministic function of
the data.

\subsection{Feature selection as part of selection}\label{sec:search}

A tabular learner is only as good as its feature list, and in the
monthly-horizon regime the list that helps is short. We therefore treat the
feature list as a hyperparameter: the search layer proposes states
$(\theta, F)$ --- a hyperparameter vector and a feature list --- and every
learned method in the comparison is searched the same way with the same
candidate set and budget.

\paragraph{Candidates.} From the raw snapshot (Section~\ref{sec:data}) we
build one large \emph{candidate set}: for every price-like series,
lookback return, drawdown, moving-average ratio, $z$-score, realized
volatility, return skewness and kurtosis, a Hurst exponent and a percentile
rank over 1-, 3-, 6-, 12- and 24-month windows; for every yield-like
series, change, moving-average slope, $z$-score, change volatility,
skewness, kurtosis and percentile rank over the same windows; and, for
every series and window, the window-quantile features of
\citet{dempster2023quant} --- quantiles of the $z$-scored window, of its
smoothed first differences and of its second differences, read over the
whole window, its recent half and its recent quarter. Every column is a
backward-looking transform of same-day closes; the hand-built 174-column
set of the no-search pipeline is included unchanged. This gives 7{,}871
columns, finite everywhere from 2003-02.

\paragraph{Ranking.} The search needs an \emph{add pool} it can draw from.
The root feature list is the 20 columns with the largest split gain when
the committed configuration is fitted on the hand-built set (what the
no-search pipeline already relied on most). Every other candidate is then
fitted \emph{once}, added to that root, and ranked by its share of the
fit's total gain --- a marginal-gain ranking, one small fit per candidate
(7{,}856 fits; minutes on 30 cores). The top 300 form the pool. Gain is a
within-fit quantity and is not the selection score; it only orders the
candidates the protocol will later judge.

\paragraph{Search.} A stochastic local search over states, each scored
once by the protocol of Section~\ref{sec:protocol}. Starting from the
committed hyperparameters with the root feature list, each iteration
proposes one local move --- perturb the hyperparameters (within the
spaces of Appendix~\ref{app:tuning}), add one pool column
(rank-weighted), or drop one low-gain column --- and re-scores the
resulting state; columns whose split gain falls to zero are pruned, so
no state carries dead weight. The search runs for a wall-clock budget
(60 minutes on 30 cores for KellyBoost, 45 for each baseline) and
commits the best-scoring state. The MLP has
no split gain; it borrows KellyBoost's root and pool and searches its own
hyperparameters and list over them. The development segment is the only
data any of this ever sees.

\subsection{Selection versus estimation}\label{sec:protocol}

We keep two questions --- \emph{which configuration is better} and
\emph{how good is the result} --- in two separate instruments, and never
read one as the other.

\paragraph{Selection.} Hyperparameters \emph{and the feature list} are
chosen together on the \emph{development segment} only (2003-02 to
2012-12), by the search of Section~\ref{sec:search}. Each configuration is
scored by a purged walk-forward: test blocks are \emph{single rows} every
11 rows; the $h-1$ rows before a test row, whose labels overlap it, are
purged from training, and 60 rows after it are embargoed
\citep{lopezdeprado2018}. Single-row test blocks are a fidelity choice: in
deployment the predicted row always sits immediately after the training
window, and wider test blocks would score rows at distances from the
training edge that never occur live. Each block is scored by a single fit
--- the mean over blocks already averages away fit noise. This score
\emph{ranks} configurations; because each block's training set contains
data from after its test row, it is near-in-sample by construction and is
never quoted as performance.

\paragraph{Estimation.} Performance is measured once, on the
\emph{evaluation segment} (2013-01 to 2026-07), by a deployment
walk-forward that only ever looks backward: decisions every 21 trading
days; at each decision the model is refit (as the
Section~\ref{sec:ensemble} ensemble) on an expanding window containing
exactly the rows whose labels are fully realized by the decision date;
parameters stay frozen at the development winners. The decision grid is
anchored at the \emph{last} decidable row and counted backwards, so the
newest information is always used; a front-anchored grid can leave the
final decision up to 20 rows stale, and grid phase alone is known to move
backtest conclusions \citep{hoffstein2020}. The evaluation segment was not
consulted at any point during development of the method or the tuning.

\section{Data}\label{sec:data}

All inputs are daily series freely downloadable from Yahoo Finance without
an API key; the raw snapshot (38 series, 1985--2026) is committed to the
repository, and everything downstream is a deterministic function of that
one file.

\paragraph{Asset legs.} Eight legs: growth equity (VIGRX), value equity
(VIVAX), long-term Treasuries (VUSTX), international equity (VGTSX),
energy equity (VGENX), gold (GC=F), silver (SI=F), and cash. Index mutual
funds are chosen over the equivalent ETFs for their longer history (VUSTX
starts in 1986 versus 2002 for TLT); their adjusted closes include
distributions, so all legs are total-return series. The cash leg pays the
13-week T-bill yield (\^{}IRX): its 20-day accrual
$\text{yield}/100 \times 20/252$ is known at decision time.

\paragraph{Feature-only series.} The S\&P~500, VIX, three Treasury yields,
copper, eight SPDR sectors, and fifteen cross-asset context series ---
the dollar index, Russell 2000, Nasdaq, emerging-market equity, REITs,
investment-grade and high-yield credit, intermediate Treasuries, TIPS,
natural gas, the CBOE SKEW index, corn, wheat, soybeans and the yen ---
again using mutual-fund share classes where the ETF is too young, so that
every series reaches back to 2001 or earlier. None is ever a leg.

\paragraph{Features.} Two feature sets appear in the paper. The
\emph{hand-built} set: 174 columns --- momentum, volatility,
moving-average ratio, drawdown and $z$-score over 21--252-day windows for
each of 17 price series, level and 21/63-day changes for four Treasury
yields and three curve slopes, and level, trend and one-year percentile for
the VIX. The \emph{candidate} set of Section~\ref{sec:search}: 7{,}871
columns, a superset of the hand-built set, from which the search selects
each method's deployed list. Both obey the same rules: there are no macro
releases, hence no publication lags, no revisions, and no interpolation
toward future values anywhere --- what the model sees on date $t$ was
public at $t$'s close. Futures series are forward-filled onto the NYSE
calendar (a causal fill). Each assembled panel is finite everywhere by
construction --- assembly fails otherwise --- and spans 2003-02 (hand-built:
2001-08) to 2026-07 once all warm-up windows are full.

\paragraph{Labels.} Per-leg forward 20-trading-day simple returns.
Adjacent rows overlap $h-1 = 19$ days of their labels; the purge in
Section~\ref{sec:protocol} exists precisely for this overlap.

\section{Experiments}\label{sec:experiments}

\subsection{Methods compared}\label{sec:methods}

Every learned method draws from the identical candidate set and label
frame, runs the identical search (Section~\ref{sec:search}) under the
identical development protocol with the budgets of
Appendix~\ref{app:tuning}, freezes its winner --- hyperparameters and
feature list --- and deploys with the identical $K_{\mathrm{ens}}=4$
leave-one-out ensemble and walk-forward.

\begin{itemize}
\item \textbf{KellyBoost} (ours): Section~\ref{sec:method}.
\item \textbf{Two-stage predict-then-optimize}: one LightGBM regressor
\citep{ke2017lightgbm} per risky leg forecasts the 20-day return; the
forecast vector $\hat\mu$ and a trailing sample covariance $\widehat\Sigma$
of the label vectors feed
$\max_{w \in \simplex} \; w^{\top}\hat\mu - \tfrac12 w^{\top}\widehat\Sigma\, w$
--- the second-order expansion of expected log growth
\citep{hakansson1971}, i.e.\ a plug-in Kelly portfolio. The covariance
window is tuned alongside the regressor's hyperparameters.
\item \textbf{Multiclass surrogate}: LightGBM multiclass on the label
``index of the best-performing leg''; the predicted class-probability
vector is itself a softmax portfolio. This is the natural classification
shortcut to the same output space.
\item \textbf{MLP, same loss}: a feed-forward network (up to two hidden
layers, Adam, weight decay) trained on exactly
loss~\eqref{eq:loss} with a softmax head --- the neural end-to-end
reference, sharing every convention including the ensemble.
\item \textbf{MLP, surrogate loss}: the identical network trained with
cross-entropy on the argmax label. With the two tree methods this
completes a $2 \times 2$ design --- \{boosted tree, MLP\} $\times$
\{growth objective, classification surrogate\} --- so the effect of the
loss can be read \emph{within} each learner class, where nothing else
differs.
\item \textbf{Unconditional Kelly}: the constant simplex portfolio
maximizing in-sample log growth, refit each decision date --- what
conditioning on features is worth is the gap to this baseline
\citep{cover1991}.
\end{itemize}

\subsection{Metrics}\label{sec:metrics}

The primary metric is the realized mean log growth per 20-day decision,
$\overline{\log G} = \frac{1}{T}\sum_t \log(1 + w_t^{\top} y_t)$ over the
evaluation decisions --- the deployed quantity the objective optimizes. We
also report, from daily equity curves with holdings drifting between
rebalances: annualized geometric return, annualized volatility, Sharpe
ratio of daily returns in excess of the T-bill leg, and maximum drawdown.
Pairwise differences in daily log growth against KellyBoost get
moving-block bootstrap 95\% intervals (block 21 days, 5{,}000 draws)
\citep{politis1994}. The $2 \times 2$ loss-effect pairs get their own
paired test at the decision level: per-decision differences in realized
log growth, circular moving-block bootstrap (block 6 decisions, 10{,}000
draws), per cell and pooled across the four aligned cells
(Table~\ref{tab:losseffect}). Backtests are gross of costs; Table~\ref{tab:costs}
prices costs back in via measured turnover.

\section{Results}\label{sec:results}

\begin{table}[t]
\centering
\caption{Out-of-sample results on the evaluation segment (2013-01 to
2026-07, 163 monthly decisions, gross of costs). $\overline{\log G}$ is the
mean realized log growth per 20-day decision ($\times 100$). Sharpe is
daily returns in excess of T-bills, annualized. The last column is the
moving-block-bootstrap 95\% interval for the difference in annualized daily
log growth versus KellyBoost (negative = worse than KellyBoost).}
\label{tab:main}
\medskip
\IfFileExists{tables/main_results.tex}{\footnotesize\input{tables/main_results}}{[run experiments]}
\end{table}

\begin{figure}[t]
\centering
\IfFileExists{figures/equity_curves.pdf}{\includegraphics[width=\textwidth]{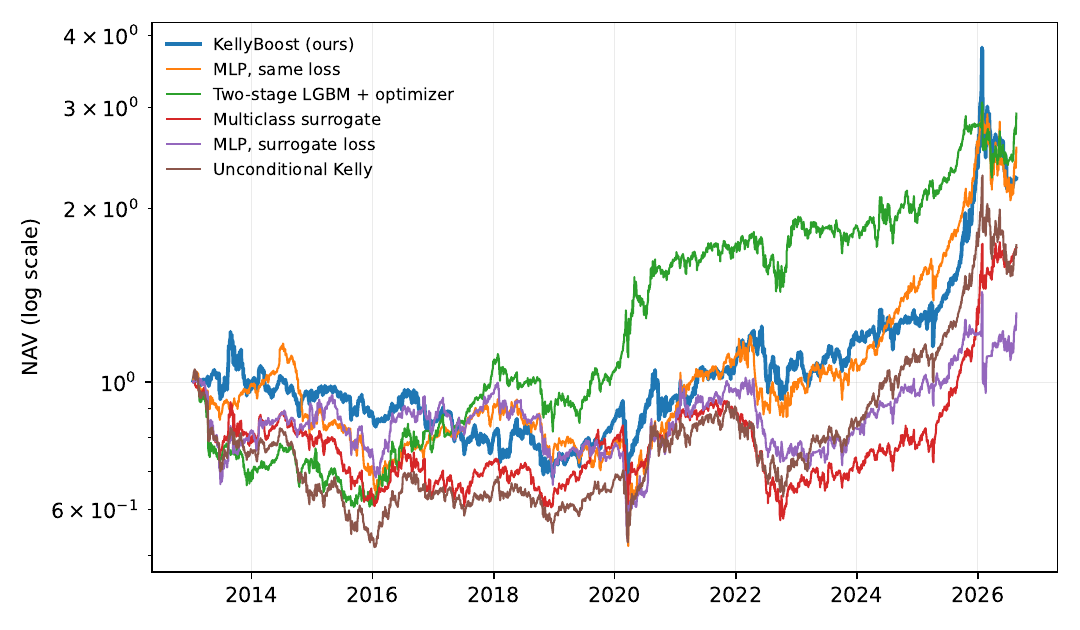}}{[run experiments]}
\caption{Out-of-sample equity curves (log scale), all methods, evaluation
segment. Every learned method is deployed identically: frozen
development-tuned parameters, expanding-window refits every 21 trading
days, $K_{\mathrm{ens}}=4$ leave-one-out ensemble.}
\label{fig:equity}
\end{figure}

\begin{figure}[t]
\centering
\IfFileExists{figures/weights.pdf}{\includegraphics[width=\textwidth]{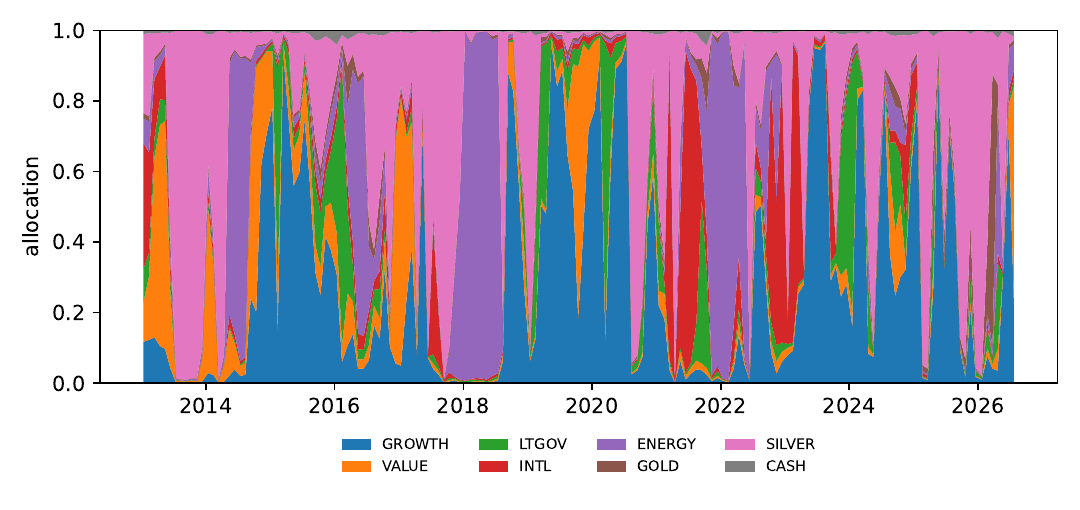}}{[run experiments]}
\caption{KellyBoost's deployed allocation over the evaluation segment
(committee average). The model was never told what an asset class is ---
the rotation structure is learned from the features alone.}
\label{fig:weights}
\end{figure}

Table~\ref{tab:main} and Figure~\ref{fig:equity} present the main
comparison; Figure~\ref{fig:weights} shows the deployed allocation. One reading rule
first: every pairwise bootstrap interval in Table~\ref{tab:main} straddles
zero. Monthly-horizon evidence over \evalyears{} years cannot separate any
two of these methods at conventional significance, so the section argues
from consistency --- across methods, and across ablations that move the
same lever from two directions --- and
reads the comparisons in two groups: first the $2 \times 2$ the paper's
claim rests on, where only the training loss moves; then the
pipeline-shape comparisons that mark the objective's boundary.

\paragraph{The loss effect: the $2 \times 2$.} Four methods share one
output space (softmax over legs) and pair up into two controlled
comparisons of the training loss: boosted trees under the growth objective
versus cross-entropy on the argmax label, and the identical MLP under the
same two losses. The growth objective wins all four cells --- trees
$0.47$ vs $0.35$ searched and $0.39$ vs $0.32$ hand-built (where both
methods see the identical 174 columns), MLPs $0.56$ vs $0.20$ searched
and $0.67$ vs $0.49$ hand-built (identical everything but the loss;
Tables~\ref{tab:main} and~\ref{tab:search}). Table~\ref{tab:losseffect}
puts the claim on its own statistical footing: the paired per-decision
difference is positive in every cell ($+0.07$ to $+0.37$ per 20-day
decision, $\times 100$), and pooling the four aligned difference series
--- an average per decision date, which preserves the cross-cell
correlation rather than pretending the cells are independent --- gives
$+0.18$ with bootstrap $\Pr(\Delta > 0) = 0.89$, stable at $0.84$--$0.96$
across bootstrap block lengths of 1 to 12 decisions. The pooled interval
still includes zero --- thirteen years of monthly decisions is a small
sample --- but the direction never flips, in any cell, under any block
length. The mechanism is not mysterious: cross-entropy on the argmax
label throws away the margin structure of returns --- a month won by 20
basis points and a month won by 20 points are the same training example
--- and both learner classes pay for discarding it. The comparison also
survives --- indeed sharpens under --- transaction costs, because the
growth objective trades less than the surrogate in every cell
(Section~\ref{sec:ablations}). (One residual
confound in the tree pair: the surrogate runs on LightGBM rather than
XGBoost, and in the searched rows each method deploys its own feature
list; the MLP pair has neither.)

\begin{table}[t]
\centering
\caption{The loss effect, cell by cell: paired per-decision differences in
realized log growth, growth objective minus argmax surrogate, holding the
learner class and feature pipeline fixed within each cell
($K_{\mathrm{ens}}=4$, 163 decisions). Intervals: circular moving-block
bootstrap, block 6 decisions, 10{,}000 draws. The pooled row averages the
four aligned difference series at each decision date, preserving
cross-cell correlation.}
\label{tab:losseffect}
\medskip
\IfFileExists{tables/loss_effect.tex}{\small\input{tables/loss_effect}}{[run experiments]}
\end{table}

\paragraph{Across learner classes: no conclusion.} At fixed loss, the MLP
edges the tree under the growth objective ($0.56$ vs $0.47$) and the tree
edges the MLP under the surrogate ($0.35$ vs $0.20$). Both gaps sit well
inside bootstrap noise and are confounded by everything a learner class is
--- capacity, optimization geometry, regularization style --- so we read
nothing into them, in either direction.

\paragraph{The objective's boundary: what shrinks, wins the window.} The
remaining comparisons change the pipeline's shape, not just its loss, and
they mark where faithfulness to the objective stops paying. The two-stage
pipeline shares KellyBoost's learner class and its decision rule's
\emph{intent} (growth optimality via the quadratic expansion) but learns
squared-error forecasts first; a monthly-horizon return is barely
predictable, so the fitted forecasts sit near zero, and an optimizer fed
near-zero expected returns with a full covariance matrix holds a mild,
diversified book. What the predict-then-optimize literature treats as the
seam's defect \citep{elmachtoub2022} functions here as implicit
regularization, and the two-stage pipeline out-deploys every end-to-end
learner. KellyBoost's side
of the boundary is symmetric: it does precisely what its objective asks,
estimating the conditional growth-optimal portfolio --- aggressive by
construction, deployed committee mean maximum leg weight $0.70$ --- and it
pays for the aggression in 2014--2018 and again in 2026, after banking the
study's largest single-year gain in 2025. The unconditional-Kelly row
completes the picture: conditioning on features is worth something ---
the conditional model beats the best constant portfolio. Getting the
decision loss right is necessary; this boundary is
the measured distance between necessary and sufficient.

\begin{table}[t]
\centering
\caption{What the feature search bought. Hand-built: the 174-column set
with random-search hyperparameters (the pipeline as first designed);
searched: Section~\ref{sec:search}. Deployment metrics, $K_{\mathrm{ens}}=4$.}
\label{tab:search}
\medskip
\IfFileExists{tables/search.tex}{\footnotesize\input{tables/search}}{[run experiments]}
\end{table}

\paragraph{The feature search improved selection scores far more than
deployment --- except where shrinkage could use it.} Table~\ref{tab:search}:
the joint search roughly doubled every method's development-segment
selection score. Deployed, the gains sort by pipeline, not by score: the
two-stage pipeline converts the searched features into a doubling of
realized growth --- better inputs make better forecasts, and its shrinkage
keeps the resulting bets survivable --- while the end-to-end methods
convert the same kind of search into little more than reshuffled
concentration, and both MLPs' searched lists actually deploy \emph{worse}
than their hand-built ones. The residual gap between selection movement and deployment
movement is selection bias made visible, which is why no selection score
in this paper is ever quoted as performance. For the record,
KellyBoost's deployed list is five columns (Table~\ref{tab:features}) ---
two-year window quantiles of international equity curvature, wheat and
copper, a soybean drift quantile, and financial-sector return skewness ---
a list whose economic legibility we leave to the reader's judgment, which
is part of the point.

\begin{table}[t]
\centering
\caption{KellyBoost's searched feature list, with each column's share of
the committed fit's total split gain.}
\label{tab:features}
\medskip
\IfFileExists{tables/features.tex}{\small\input{tables/features}}{[run experiments]}
\end{table}

\section{Ablations}\label{sec:ablations}

\begin{table}[t]
\centering
\caption{Ablations. Dev.\ score: mean log growth ($\times 100$) under the
development tuning protocol at the frozen KellyBoost parameters.
Deployment columns as in Table~\ref{tab:main}, all with
$K_{\mathrm{ens}}=4$.}
\label{tab:ablations}
\medskip
\IfFileExists{tables/ablations.tex}{\small\input{tables/ablations}}{[run experiments]}
\end{table}

\begin{table}[t]
\centering
\caption{Transaction-cost sensitivity: annualized return / excess Sharpe
under one-way proportional costs applied to each method's measured
turnover.}
\label{tab:costs}
\medskip
\IfFileExists{tables/costs.tex}{\footnotesize\input{tables/costs}}{[run experiments]}
\end{table}

\paragraph{The exact Hessian is a better optimizer and a worse deployer.}
Replacing $|h_k|$ from \eqref{eq:gradhess} with a constant (gradient-only
boosting) or with $g_k^2$ changes nothing about the objective's minimum and
everything about the optimization path. On the development protocol at
identical hyperparameters the exact curvature scores $+4.19$ against the
constant's $+1.43$ ($\times 100$; Table~\ref{tab:ablations}) --- it is,
as second-order theory promises, much the better optimizer of the training
objective. Deployed, the ranking inverts: the constant-Hessian variant's
undersized steps leave its softmax closer to uniform, and that accidental
shrinkage beats the faithfully optimized full-Kelly allocation out of
sample. The same inversion, from the other side: the $g_k^2$ surrogate is
the worst selector and roughly matches the exact Hessian's deployment. The
exact curvature is not decorative --- it does exactly what it claims ---
but on this problem, what it claims is the wrong thing to want.

\paragraph{Vector leaves.} Growing $K$ independent trees per round
(\texttt{one\_output\_per\_tree}) removes the shared partition and
multiplies parameters by $K$. It selects worse and deploys better
(Table~\ref{tab:ablations}) --- the same shrinkage signature as the
Hessian substitutions, since $K$ uncoordinated trees move the softmax less
coherently than one vector-leaf tree.

\paragraph{Ensembling stabilizes the weights and buys no growth.}
Prefix averages of the stored members (no refits) confirm the single-fit
instability claimed in Section~\ref{sec:ensemble} --- one-row perturbations
move deployed legs by whole percentage points, and the committee's gap to
the $K{=}16$ book shrinks at roughly $1/\sqrt{K}$ --- while realized
growth is flat in $K$ to within noise. Averaging removes an operational
fragility (the deployed book no longer depends on a split-tie coin flip),
not a performance penalty; we keep $K_{\mathrm{ens}} = 4$ for determinism,
and report this honestly rather than as an accuracy ingredient. Because
members differ by one training row each, the same numbers are
simultaneously the measurement of that instability.

\paragraph{Off-diagonal curvature: the end-to-end check.} The pure-numpy
engine's full-Hessian leaf solver
(\texttt{leaf\_solver="full"}) solves every leaf against
the aggregated $K \times K$ row Hessians of \eqref{eq:fullhess}
(eigenvalue-rectified) instead of the coordinate-wise step, with the split
search held fixed. Under the tuning protocol at the frozen parameters, on
one identical coarsened block grid, the diagonal solver scores $+0.0335$
and the full solver $+0.0339$ (mean log growth per block) --- within
noise --- consistent with the
$\geq 0.99$ leaf-step cosines of Section~\ref{sec:curvature}: at the
tuned regularization, the off-diagonal terms buy nothing that the
diagonal step had lost.

\paragraph{Costs.} Table~\ref{tab:costs} prices one-way proportional
costs into every method through its measured turnover; at realistic
ETF/fund cost levels the conclusions of Table~\ref{tab:main} are
unchanged. More important for the paper's claim: in all four cells of
the $2 \times 2$ the growth objective trades \emph{less} than its
argmax counterpart --- $7.9\times$ versus $13.1\times$ annualized
turnover for the searched trees, $4.8\times$ versus $10.8\times$
hand-built, $11.8\times$ versus $16.8\times$ and $15.4\times$ versus
$16.0\times$ for the MLPs --- so charging costs \emph{widens} the loss
effect: the pooled paired difference of Table~\ref{tab:losseffect}
rises from $+0.18$ gross to $+0.20$, $+0.22$ and $+0.25$ per decision
at 5, 10 and 20 bps, with $\Pr(\Delta > 0)$ rising from $0.89$ to
$0.95$. The reading we offer: the argmax label is the more brittle
target --- a hair's-breadth change in which leg wins the month flips
the entire label, and the classifier's book follows it --- while the
growth objective moves its weights only as far as estimated conditional
growth moves.

\section{Limitations}\label{sec:limitations}

The study is one universe, two feature pipelines, one 13.6-year
evaluation window; the protocol prevents tuning-set leakage but cannot manufacture
more independent history --- monthly-horizon evidence is intrinsically
scarce relative to daily-horizon studies. The formulation is long-only and
fully invested (leverage and shorting would enter through an affine map of
the softmax, which we have not evaluated). Labels overlap, which the purge
handles in tuning but which still reduces the effective sample size of any
metric computed on daily curves; our bootstrap blocks are sized to the
overlap. Backtests are gross of costs in the main table, with sensitivity
priced separately; taxes, market impact and capacity are out of scope.

One disclosure belongs here rather than in fine print. The evaluation
segment was scored twice, not once. The pipeline as first designed ---
hand-built features, random-search hyperparameters --- was evaluated,
underperformed, and the feature-search layer of
Section~\ref{sec:search} was added \emph{in response}. The search itself
never touched the evaluation segment, but the decision to build it was
informed by an evaluation-segment result, and a reader should discount the
searched rows of Table~\ref{tab:search} accordingly. Both pipelines are
reported in full, the first look is the hand-built row, and no further
design iteration followed the second look. The paper's central negative
finding is robust to this: it holds in both pipelines, and the disclosed
peek could only have biased the study \emph{toward} the positive result it
failed to find.

Two further limitations deserve their own sentences. First, the growth-optimal
criterion is the aggressive end of the risk spectrum, and estimation error
pushes a deployed full-Kelly policy toward over-betting
\citep{maclean2011}; we treat this as a dial rather than a defect ---
a fractional-Kelly blend of the deployed weights, or the CRRA objective
of Section~\ref{sec:objective}, tempers it --- but choosing the
dial's setting is a preference, not a statistical question this paper
answers. Second, scale: the per-round cost is linear in $K$, but a
cross-sectional problem with hundreds of assets is a different design
point --- there, weights of $O(1/K)$ leave the softmax's cells
individually tiny, and the natural extensions (asset-symmetric shared
features, grouped or hierarchical softmax over sectors, sparse top-$m$
allocation) change the architecture around the objective while the
objective, its derivatives and Proposition~\ref{prop:kelly} carry over
unchanged. We claim the asset-allocation regime, not the stock-selection
one. Finally, nothing in this paper is investment advice.

\section{Conclusion}\label{sec:conclusion}

End-to-end portfolio learning required backpropagation; boosted trees ---
the learner tabular practitioners actually reach for --- could join it
only through surrogates. KellyBoost closes that gap with nothing
approximate: softmax leaves on the simplex, loss equal to negative log
growth, gradient and curvature in closed form, a dependency-free reference
engine, and finite-difference tests that pin every derivative. The claim
the experiments then support is deliberately narrow and, within its
bounds, clean: swapping the classification surrogate for the exact
decision objective improves deployed growth \emph{within} both learner
classes we tried, in both feature pipelines --- four cells, one direction.
The same testbed prices the objective's boundary with equal candor. An
estimated conditional full-Kelly allocation concentrates beyond what
monthly-horizon signals can support, so the end-to-end learners ---
whichever their architecture --- trail a two-stage pipeline whose
squared-error stage shrinks forecasts toward zero; the exact Hessian, demonstrably the better
optimizer on the selection protocol, loses deployed to its own de-tuned
substitutes. \emph{What} the decision loss asks for dominates \emph{how
well} it is optimized. The constructive reading, which we leave to future
work: keep the exact objective, and put the shrinkage \emph{into} it on
purpose --- distributionally robust or explicitly regularized growth
objectives, turnover-penalized variants (differentiable in $z$), and a
risk dial calibrated by criteria the development segment can actually
price --- rather than receiving it as a happy accident of a misspecified
pipeline.

\bibliographystyle{plainnat}
\bibliography{references}

\appendix

\section{Derivation of the gradient and Hessian}\label{app:derivation}

Fix one row and drop $t$. With $\sigma = \mathrm{softmax}(z)$,
$S = \sum_j \sigma_j y_j$ and $\ell = -\log(1+S)$:
\[
\frac{\partial \sigma_j}{\partial z_k} = \sigma_j(\delta_{jk} - \sigma_k)
\;\;\Longrightarrow\;\;
\frac{\partial S}{\partial z_k}
= \sum_j y_j \sigma_j (\delta_{jk} - \sigma_k)
= \sigma_k y_k - \sigma_k S
= \sigma_k (y_k - S).
\]
Hence
\[
g_k = \frac{\partial \ell}{\partial z_k}
= -\frac{1}{1+S}\,\frac{\partial S}{\partial z_k}
= \frac{\sigma_k (S - y_k)}{1+S}.
\]
For the diagonal second derivative, differentiate $g_k$ once more:
\[
\frac{\partial g_k}{\partial z_k}
= \underbrace{\frac{\sigma_k(1-\sigma_k)(S-y_k)}{1+S}}_{\text{from } \sigma_k}
+ \underbrace{\frac{\sigma_k\,\sigma_k(y_k-S)}{1+S}}_{\text{from } S \text{ in the numerator}}
- \underbrace{\frac{\sigma_k(S-y_k)\,\sigma_k(y_k-S)}{(1+S)^2}}_{\text{from } S \text{ in the denominator}} .
\]
The first two terms combine to
$\sigma_k(S-y_k)(1-2\sigma_k)/(1+S) = g_k(1-2\sigma_k)$; the third equals
$+\sigma_k^2 (S-y_k)^2/(1+S)^2 = g_k^2$. Therefore
$h_k = g_k(1-2\sigma_k) + g_k^2$, as in \eqref{eq:gradhess}. Both $g$ and
$h$ are verified against central finite differences (relative tolerance
$10^{-5}$ and $10^{-3}$ respectively) in the accompanying tests.

\paragraph{The full Hessian.} Write $a_k = \sigma_k(y_k - S) = \partial S
/ \partial z_k$ and $\ell = \phi(S)$ for a general twice-differentiable
$\phi$ (the log loss is $\phi(S) = -\log(1+S)$). Differentiating $a_k$:
\[
\frac{\partial a_k}{\partial z_j}
= \sigma_k(\delta_{jk} - \sigma_j)(y_k - S) - \sigma_k\,\frac{\partial S}{\partial z_j}
= \delta_{jk}\, a_k - \sigma_j a_k - \sigma_k a_j ,
\]
which is symmetric in $(j,k)$; in matrix form $\partial a / \partial z =
\mathrm{diag}(a) - \sigma a^{\top} - a \sigma^{\top}$. The chain rule then
gives the per-row Hessian of \eqref{eq:fullhess}:
\[
\frac{\partial^2 \ell}{\partial z \,\partial z^{\top}}
= \phi''(S)\, a a^{\top}
+ \phi'(S)\left(\mathrm{diag}(a) - \sigma a^{\top} - a \sigma^{\top}\right).
\]
Its diagonal recovers $h_k = \phi'' a_k^2 + \phi' a_k (1 - 2\sigma_k)$,
which for the log loss is \eqref{eq:gradhess}. The identity is verified
entry-by-entry against second-order central finite differences in the
accompanying tests, including its symmetry.

\paragraph{The CRRA family.} For
$\phi_{\gamma}(S) = \big((1+S)^{1-\gamma} - 1\big)/(\gamma-1)$,
$\gamma \neq 1$:
\[
\phi_{\gamma}'(S) = -(1+S)^{-\gamma},
\qquad
\phi_{\gamma}''(S) = \gamma\, (1+S)^{-\gamma-1},
\]
so $g_k = \phi_{\gamma}' a_k$ and $h_k = \phi_{\gamma}'' a_k^2 +
g_k(1-2\sigma_k)$ drop into every formula above unchanged;
$\gamma \to 1$ recovers the log-loss derivatives continuously. The
gradient and the $\gamma \to 1$ limit are covered by finite-difference
tests as well.

\section{Search spaces and budgets}\label{app:tuning}

Every method is searched by the procedure of Section~\ref{sec:search}
with a fixed seed, on the identical development protocol (single-row
purged blocks every 11 rows; every 63 rows for the MLP, whose fit is an
order of magnitude slower --- selection scores are never compared across
methods, only across a method's own states). Wall-clock budgets on 30
cores: KellyBoost 60 minutes; the two LightGBM pipelines and the MLP 45
minutes each. Pool size 300, root 20 columns, at most 40 columns per
state. Parameter moves multiply every tunable by a factor in
$[1/1.25, 1.25]$ and clamp to the following spaces: KellyBoost --- rounds
$[60, 600]$, $\eta$ $[0.01, 0.5]$, L1 $[10^{-4}, 10]$, L2 $[0.01, 20]$,
$\gamma$ $[10^{-4}, 5]$, \texttt{min\_child\_weight} $[0.1, 100]$, leaves
$[7, 63]$ (loss-guided growth), CRRA risk aversion $[1, 10]$
(Section~\ref{sec:objective}). LightGBM --- trees $[40, 400]$, learning
rate $[0.01, 0.4]$, leaves $[4, 63]$, \texttt{min\_child\_samples}
$[5, 120]$, L1 $[10^{-4}, 5]$, L2 $[10^{-4}, 30]$, column subsample
$[0.3, 1.0]$, and for the two-stage pipeline a covariance window
$[500, 3000]$ rows. MLP --- hidden $(128, 64)$, learning rate
$[10^{-4}, 10^{-2}]$, weight decay $[10^{-6}, 0.1]$, epochs $[40, 240]$,
batch 256. The root hyperparameters are the winners of the earlier
random search over the hand-built set (40 / 12 / 16 / 16 trials; kept in
the repository for the record). Winning configurations and feature lists
are committed as JSON; the full search logs are kept under
\texttt{results/}.

\section{Reproduction}\label{app:repro}

\texttt{uv sync \&\& uv run pytest \&\& bash run\_all.sh} regenerates every
number, table and figure in this paper from the committed data snapshot
(fixed seeds for every fit; the searches are wall-clock budgeted, so their
exact trees are hardware-dependent --- the committed winners under
\texttt{experiments/params/} are the record the tables regenerate from;
roughly six hours on 30 CPU cores). Refreshing
the snapshot itself is one further command and requires no API key.

\end{document}

%% file: tables/main_results.tex
\begin{tabular}{lrrrrrr}
\toprule
 & $\overline{\log G}$ & \multicolumn{1}{c}{Ann.} & \multicolumn{1}{c}{Ann.} &  & \multicolumn{1}{c}{Max} & $\Delta$ ann.\ log growth \\
Method & ($\times 100$, 20d) & \multicolumn{1}{c}{return} & \multicolumn{1}{c}{vol.} & Sharpe & \multicolumn{1}{c}{DD} & vs.\ KellyBoost (\%) \\
\midrule
KellyBoost (ours) & 0.47 & 6.2\% & 23.9\% & 0.30 & -46.2\% & --- \\
MLP, same loss & 0.56 & 7.1\% & 20.5\% & 0.35 & -55.5\% & -0.9 [-9.2, +8.9] \\
Two-stage LGBM + optimizer & 0.82 & 8.2\% & 20.9\% & 0.40 & -40.4\% & -1.9 [-12.3, +10.0] \\
Multiclass surrogate & 0.35 & 4.0\% & 19.2\% & 0.21 & -42.7\% & +2.0 [-5.6, +9.4] \\
MLP, surrogate loss & 0.20 & 2.0\% & 21.8\% & 0.12 & -47.4\% & +4.0 [-5.1, +14.2] \\
Unconditional Kelly & 0.42 & 4.1\% & 19.3\% & 0.22 & -50.9\% & +2.0 [-4.6, +9.3] \\
\bottomrule
\end{tabular}

%% file: tables/loss_effect.tex
\begin{tabular}{lrrr}
\toprule
 & $\Delta\overline{\log G}$ & & \\
Cell (growth $-$ surrogate) & ($\times 100$, 20d) & 95\% CI & $\Pr(\Delta > 0)$ \\
\midrule
Boosted tree, searched & +0.12 & [-0.52, +0.75] & 0.65 \\
Boosted tree, hand-built & +0.07 & [-0.34, +0.45] & 0.63 \\
MLP, searched & +0.37 & [-0.26, +0.99] & 0.87 \\
MLP, hand-built & +0.18 & [-0.32, +0.75] & 0.77 \\
\midrule
Pooled (four cells) & +0.18 & [-0.11, +0.50] & 0.89 \\
\bottomrule
\end{tabular}

%% file: tables/search.tex
\begin{tabular}{llrrrrrr}
\toprule
 & & & $\overline{\log G}$ & Ann.\ & Ann.\ & & Max \\
Method & Features & $|F|$ & ($\times 100$) & return & vol. & Sharpe & DD \\
\midrule
KellyBoost & hand-built & 174 & 0.39 & 5.7\% & 22.9\% & 0.28 & -48.3\% \\
KellyBoost & searched & 5 & 0.47 & 6.2\% & 23.9\% & 0.30 & -46.2\% \\
\addlinespace
MLP, same loss & hand-built & 174 & 0.67 & 8.5\% & 19.4\% & 0.43 & -40.4\% \\
MLP, same loss & searched & 19 & 0.56 & 7.1\% & 20.5\% & 0.35 & -55.5\% \\
\addlinespace
Two-stage LGBM + optimizer & hand-built & 174 & 0.38 & 5.2\% & 23.8\% & 0.26 & -36.2\% \\
Two-stage LGBM + optimizer & searched & 15 & 0.82 & 8.2\% & 20.9\% & 0.40 & -40.4\% \\
\addlinespace
Multiclass surrogate & hand-built & 174 & 0.32 & 4.5\% & 20.1\% & 0.23 & -41.3\% \\
Multiclass surrogate & searched & 26 & 0.35 & 4.0\% & 19.2\% & 0.21 & -42.7\% \\
\addlinespace
MLP, surrogate loss & hand-built & 174 & 0.49 & 6.7\% & 19.5\% & 0.34 & -33.1\% \\
MLP, surrogate loss & searched & 23 & 0.20 & 2.0\% & 21.8\% & 0.12 & -47.4\% \\
\bottomrule
\end{tabular}

%% file: tables/features.tex
\begin{tabular}{lr}
\toprule
Feature & Gain share \\
\midrule
\texttt{INTL\_Q2Y\_SH2P90} & 35.9\% \\
\texttt{WHEAT\_Q2Y\_RH2P25} & 27.7\% \\
\texttt{COPPER\_Q2Y\_SWP10} & 22.0\% \\
\texttt{SOYBEAN\_Q2Y\_RH2P50} & 7.7\% \\
\texttt{XLF\_SKEW1Y} & 6.7\% \\
\bottomrule
\end{tabular}

%% file: tables/ablations.tex
\begin{tabular}{lrrrrr}
\toprule
 & Dev.\ score & $\overline{\log G}$ & Ann.\ & & Max \\
Variant & ($\times 100$) & ($\times 100$) & return & Sharpe & DD \\
\midrule
analytic Hessian (model) & 4.19 & 0.47 & 6.2\% & 0.30 & -46.2\% \\
constant Hessian & 1.43 & 0.67 & 8.3\% & 0.63 & -23.6\% \\
$g^2$ Hessian & 0.54 & 0.45 & 5.0\% & 0.26 & -51.7\% \\
one output per tree & 3.50 & 0.65 & 8.8\% & 0.40 & -45.7\% \\
\bottomrule
\end{tabular}

%% file: tables/costs.tex
\begin{tabular}{lrrrrr}
\toprule
Method & Turnover & 0 bps & 5 bps & 10 bps & 20 bps \\
\midrule
KellyBoost (ours) & 7.9x & 6.2\% / 0.30 & 5.8\% / 0.28 & 5.3\% / 0.27 & 4.5\% / 0.23 \\
MLP, same loss & 11.8x & 7.1\% / 0.35 & 6.5\% / 0.32 & 5.9\% / 0.30 & 4.6\% / 0.24 \\
Two-stage LGBM + optimizer & 16.5x & 8.2\% / 0.40 & 7.3\% / 0.36 & 6.4\% / 0.32 & 4.7\% / 0.24 \\
Multiclass surrogate & 13.1x & 4.0\% / 0.21 & 3.3\% / 0.18 & 2.7\% / 0.14 & 1.3\% / 0.07 \\
MLP, surrogate loss & 16.8x & 2.0\% / 0.12 & 1.2\% / 0.09 & 0.3\% / 0.05 & -1.4\% / -0.03 \\
Unconditional Kelly & 1.8x & 4.1\% / 0.22 & 4.0\% / 0.21 & 3.9\% / 0.21 & 3.7\% / 0.20 \\
\bottomrule
\end{tabular}